\documentclass[3p]{elsarticle}

\usepackage[utf8]{inputenc}
\usepackage{amsfonts}
\usepackage{amsmath}
\usepackage{amsthm}
\usepackage{amscd}
\usepackage{epsfig}
\usepackage{epstopdf}
\usepackage{graphicx}
  \graphicspath{{/}{figures/}}
  \DeclareGraphicsExtensions{.pdf,.png}
\usepackage{comment}
\usepackage{enumitem}
\usepackage{xfrac}
\usepackage{tikz-cd}
\usepackage{txfonts}
\usepackage{siunitx}
\usepackage{hyperref}

\newtheorem{theorem}{Theorem}
\newtheorem{rem}{Remarks}

\begin{document}
	
\begin{frontmatter}

\title{Introduction of the \texorpdfstring{G$_2$}{G2}-Ricci Flow: Geometric Implications for Spontaneous Symmetry Breaking and Gauge Boson Masses}

\author[1]{Richard Pinčák}
\address[1]{Institute of Experimental Physics, Slovak Academy of Sciences, Watsonova 47, 043 53 Košice, Slovak Republic}
\ead{pincak@saske.sk}

\author[2]{Alexander Pigazzini}
\address[2]{Mathematical and Physical Science Foundation, 4200 Slagelse, Denmark}
\ead{pigazzinialexander18@gmail.com}

\author[1]{Michal Pudlák}
\ead{pudlak@saske.sk}

\author[4]{Erik Bartoš}%
\address[4]{Institute of Physics, Slovak Academy of Sciences, Dúbravská cesta 9, 845 11 Bratislava, Slovak Republic}%
\ead{erik.bartos@savba.sk}

\begin{abstract}
This work introduces the G$_{2}$-Ricci flow on seven-dimensional manifolds
with non-zero torsion and explores its physical implications. By extending
the Ricci flow to manifolds with G$_{2}$ structures, we study the evolution
of solitonic solutions and their role in spontaneous symmetry breaking
in gauge theories. In particular, this model proposes that the masses of
the W and Z bosons are determined not by an external scalar field, as in
the Higgs mechanism, but by the intrinsic geometric torsion of the manifold.
Furthermore, a possible connection between the geometry of extra dimensions
and the curvature of our spacetime is explored, with implications for the
experimentally observed positive cosmological constant. This approach provides
an innovative interpretation of fundamental interactions in theoretical
physics, opening new possibilities for studying extra dimensions and the
geometry of G$_{2}$-manifolds.%
\end{abstract}
%
%
%
%
%
\end{frontmatter}
%
\section{Introduction}
\label{sec1}

The Ricci flow has been a powerful tool in differential geometry, most
notably used in the proof of the Poincar\'{e} conjecture by Perelman
\cite{Perelman:2002}. The Ricci flow is a geometric evolution equation
that smooths out the metric of a manifold, guiding it toward its canonical
geometric structure. Originally introduced by Hamilton
\cite{Hamilton:1982}, the Ricci flow has since found applications in various
areas of mathematics and physics, including high-energy theoretical physics
and geometry.

Extending the concept of Ricci flow to manifolds with special holonomy
groups, the G$_{2}$-Ricci flow naturally arises in the study of 7-dimensional
manifolds endowed with G$_{2}$-structures. These structures are not only
mathematically interesting due to their unique holonomy properties, but
also important in theoretical physics, particularly in string theory, where
G$_{2}$-manifolds play a central role in compactifications
\cite{Hitchin:2001}. Witten \cite{Witten:1996}, highlighted the importance
of G$_{2}$-structures in string theory compactifications, showing how they
preserve a portion of supersymmetry, making them critical in models of
particle physics that emerge from string theory.

A G$_{2}$-structure on a 7-dimensional manifold is characterized by a 3-form
$\varphi $, which reduces the structure group to the exceptional Lie group
G$_{2}$. When $\varphi $ is both closed and co-closed, the structure is
torsion-free, and the associated metric is Ricci-flat
\cite{Joyce:2000}. However, in many physical contexts, such as sigma models
and gauge theories, it is natural to consider G$_{2}$-structures with non-zero
torsion, which complicates the evolution of these structures under the
G$_{2}$-Ricci flow \cite{Fernandez:1982}. Bryant \cite{Bryant:1987}, contributed
significantly by providing examples of G$_{2}$-structures with torsion
and exploring their geometric properties, which has been influential in
understanding how torsion affects the evolution under the G$_{2}$-Ricci
flow.

Previously, the discrete gravitational field and spacetime with canonical
defect was considered \cite{Pincak:2021}. In the Pigazzini-Pin\v{c}ak braneworld
scenario, there exists a class of manifolds in the context of D-branes
using partially negative-dimensional product, so-called PNDP manifolds
(PNDP), which have virtual dimensions. The topological defects can emerge
from non-orientability and they can be related to the appearance of curvature
at low-energy scales.

In this paper, we investigate the G$_{2}$-Ricci flow on 7-dimensional manifolds
with G$_{2}$-structures, focusing particularly on solutions involving torsion.
We rigorously analyze the conditions for existence, regularity, and convergence
of solutions. Additionally, we explore how the solitonic solutions derived
from the G$_{2}$-Ricci flow influence spontaneous symmetry breaking in
gauge theories. These solitons induce symmetry breaking, leading to the
generation of massive gauge bosons, introducing a novel geometric mechanism
that competes with the Higgs mechanism.

One of the fundamental problems in theoretical physics is to understand
the origin of the masses of gauge bosons and the role of geometry in the
fundamental interactions. In the Standard Model, the mass of the $W$ and
$Z$ bosons arises through the Higgs mechanism, but this requires the ad
hoc introduction of a scalar potential. Our approach proposes an alternative
based on the geometric torsion of G$_{2}$-manifolds, which could offer
a more natural explanation of the spontaneous electroweak symmetry breaking.
This raises several fundamental questions: Is it possible to derive symmetry
breaking directly from geometry? What are the properties of solitonic solutions
of the G$_{2}$ Ricci flow? What are the physical consequences of a purely
geometric mechanism?

This provides a deeper understanding of the interplay between geometry
and high-energy physics \cite{Karigiannis:2008}.

\section{Preliminaries}
\label{sec2}

\subsection{\texorpdfstring{G$_{2}$}{G2}-structures on 7-dimensional manifolds}
\label{sec2.1}

A 7-dimensional differentiable manifold $M$ is said to possess a G$_{2}$-structure
if there exists a 3-form $\varphi $ that reduces the structure group of
$M$ to G$_{2}$, a subgroup of SO(7). This form $\varphi $ not only defines
a Riemannian metric $g$ but also determines a corresponding 4-form
$\ast \varphi $ via relations that characterize the G$_{2}$ group. Specifically,
in a local orthonormal basis $\{e^{1}, \dots , e^{7}\}$ of the cotangent
bundle, the 3-form $\varphi $ is expressed as
%
\begin{equation}
\varphi = e^{123} + e^{145} + e^{167} + e^{246} - e^{257} - e^{347} - e^{356},
\label{eq1}
\end{equation}
where $e^{ijk} = e^{i} \wedge e^{j} \wedge e^{k}$. This 3-form
$\varphi $ directly determines a Riemannian metric $g_{\varphi}$ on
$M$ through a relation between $\varphi $ and the metric coefficients.

A G$_{2}$-structure is torsion-free if $\varphi $ is both closed and co-closed,
satisfying the conditions
%
\begin{equation}
d\varphi = 0 \quad \text{and} \quad d\ast \varphi = 0,
\label{eq2}
\end{equation}
where $\ast \varphi $ is the Hodge dual of $\varphi $. When these conditions
are met, the manifold $M_{7}$ has a torsion-free G$_{2}$-structure, implying
that the connection associated with the metric is compatible with the G$_{2}$-structure
and the manifold admits a parallel G$_{2}$-form.

\subsection{Function spaces}
\label{sec2.2}

To study the G$_{2}$-Ricci flow, we utilize Sobolev spaces
$H^{k,\alpha}(M)$, where $k$ denotes the degree of differentiability and
$\alpha $ represents the H\"{o}lder exponent. For initial conditions, the
3-form $\varphi $ is assumed to belong to $H^{k,\alpha}(M)$ for
$k \geq 2$, ensuring sufficient regularity to apply elliptic theory.

\bigskip\noindent\textit{G$_{2}$-manifolds and Sobolev spaces}

A G$_{2}$-structure is said to belong to the Sobolev space
$H^{k,\alpha}(M)$ if the components of $\varphi $ and the associated metric
$g_{\varphi}$ exhibit appropriate regularity, meaning that derivatives
up to order $k$ are H\"{o}lder continuous with exponent $\alpha $. This
ensures that the manifold is compatible with the analysis required for
the G$_{2}$-Ricci flow.

\subsection{Definition of the \texorpdfstring{G$_{2}$}{G2}-Ricci flow}
\label{sec2.3}

The G$_{2}$-Ricci flow is defined by the following equation
%
\begin{equation}
\frac{\partial \varphi}{\partial t} = \Delta _{d} \varphi +
\mathcal{L}_{X} \varphi + \mathrm{Ric} \lrcorner \ast \varphi + T(\varphi ),
\label{eq3}
\end{equation}
where
\begin{itemize}
\item $\Delta _{d}$ is the Hodge-de Rham Laplacian, a second-order elliptic
operator that acting on the 3-form $\varphi $.
\item $\mathcal{L}_{X} \varphi $ is the Lie derivative of $\varphi $ along
a vector field $X$. It is first-order operator.
\item $(\mathrm{Ric} \lrcorner \ast \varphi) $ is the contraction of the
Ricci tensor with the 4-form $\ast \varphi $.
\item $T(\varphi )$ represents the torsion of the G$_{2}$-structure, which measures the deviations from the torsion-free condition. It is decomposed according to the torsion classes  \(\tau_0, \tau_1, \tau_2, \tau_3\), (see Section 6.2, equation \ref{eq45}).
\end{itemize}

\subsection{Contraction of the Ricci tensor with \texorpdfstring{$\ast \varphi $}{*varphi}}
\label{sec2.4}

The contraction $(\mathrm{Ric} \lrcorner \ast \varphi) $ is given by
%
\begin{equation}
(\mathrm{Ric} \lrcorner \ast \varphi)_{ijk} = R_{im} (\ast \varphi)^{m}{}_{jk},
\label{eq4}
\end{equation}
where $\ast \varphi $ is the Hodge dual of $\varphi $, and the brackets
denote cyclic summation over indices. This term accounts for how the Ricci
curvature of the manifold interacts with the evolving G$_{2}$-structure
during the flow.

The contraction $R^{m}_{[i} (\ast \varphi )_{|m|jk]}$ ensures that the
resulting tensor is a well-defined 3-form. The operation removes one index
via contraction with the Ricci tensor $R^{m}_{i}$, while the symmetrization
maintains compatibility with the antisymmetry of $\ast \varphi $. This
guarantees that $(\mathrm{Ric} \lrcorner \ast \varphi) $ is consistent with
the structure of differential forms in G$_{2}$-geometry.

\section{Problem formulation}
\label{sec3}

We consider a 7-dimensional compact manifold $M$ with an initial G$_{2}$-structure
$\varphi (0) \in H^{k,\alpha}(M)$, with $k \geq 2$. Our goal is to study
the evolution of the G$_{2}$-structure under the G$_{2}$-Ricci flow, described
by the equation
%
\begin{equation}
\frac{\partial \varphi}{\partial t} = \Delta _{d} \varphi +
\mathcal{L}_{X} \varphi + \mathrm{Ric} \lrcorner \ast \varphi + T(
\varphi ).
\label{eq5}
\end{equation}
The main goal is to determine the evolution of $\varphi (t)$ under the
G$_{2}$-Ricci flow, establish the existence and uniqueness of short-time
solutions, and study the long-term behavior. In particular, we are interested
in understanding under which conditions the flow
\begin{enumerate}
\item remains smooth and convergent in the long term,
\item converges to a stationary solution (soliton) depending on the evolution
of the torsion $T(\varphi )$.
\end{enumerate}
We want to show that there are two possible scenarios for the long-term
behavior:
\begin{enumerate}
\item \textit{Ricci Solitons}: If the torsion $T(\varphi )$ decreases sufficiently
over time, the solution $\varphi (t)$ converges to a Ricci soliton.
\item \textit{Mixed Solitons}: If the torsion $T(\varphi )$ stabilizes at
a finite value, the solution $\varphi (t)$ converges to a mixed soliton,
with non-zero residual torsion.
\end{enumerate}
We assume that the initial shape $\varphi (0)$ belongs to a Sobolev space
$H^{k,\alpha}(M)$ with $k \geq 2$, which guarantees that the manifold and
the associated metric have sufficient regularity to apply the results of
elliptic theory. This regularity is necessary to guarantee the existence
of regular local solutions for the parabolic flow.

The G$_{2}$-Ricci flow problem is formulated as follows: given an initial
G$_{2}$-structure $\varphi (0)$, study the time evolution of the
$\varphi (t)$ structure according to the flow equation and determine the
conditions that guarantee the convergence of the solution to a stationary
configuration. Possible stationary configurations can include
\begin{itemize}
\item \textit{Ricci solitons}: If the torsion tends to zero, the final solution
will be a Ricci soliton, which satisfies the torsion-free flow equation.
\item \textit{Mixed Solitons}: If the torsion stabilizes at a finite value,
the final solution will be a mixed soliton, which includes a residual torsion.
\end{itemize}
The key questions we will address are
\begin{enumerate}
\item \textit{Short-Term Existence and Uniqueness}: Verify that there exists
a unique and regular solution for the G$_{2}$-Ricci flow for short times,
using the DeTurck trick to obtain a strictly parabolic shape.
\item \textit{Long-Term Behavior}: Study the evolution of the torsion
$T(\varphi )$. If $T(\varphi )$ decreases, we expect the solution to converge
to a Ricci soliton. Alternatively, if the torsion stabilizes at a non-zero
value, the solution will converge to a mixed soliton.
\item \textit{Associated Energy}: Define an energy functional that controls
the evolution of the flow and that allows us to study the stability of
the solutions in the long term.
\end{enumerate}

The theorem that we will prove in the next section, is formulated as follows
\begin{itemize}
\item \textit{Short-Term Existence and Uniqueness}: There exists a time
$T > 0$ such that the G$_{2}$-Ricci flow has a unique and regular solution
for $t \in [0, T)$, assuming that $\varphi (0) \in H^{k,\alpha}(M)$.
\item \textit{Long-Term Convergence (Ricci Soliton)}: If the torsion
$T(\varphi (t)) \to 0$ as $t \to \infty $, then the solution
$\varphi (t)$ converges to a Ricci soliton.
\item \textit{Long-Term Convergence (Mixed Soliton)}: If the torsion
$T(\varphi (t))$ stabilizes at a finite value $T_{\infty}$, then the solution
$\varphi (t)$ converges to a mixed soliton with non-zero torsion.
\end{itemize}

This formulation will allow us to explore how the torsion affects the geometric
evolution of the manifold and to understand the asymptotic behavior of
the flow.

\section{Main theorem}
\label{sec4}

\begin{theorem}%
\label{theorem:1}
Let $M$ be a compact, 7-dimensional Riemannian manifold equipped with an
initial G$_{2}$-structure $\varphi (0)$ that belongs to the Sobolev space
$H^{k,\alpha}(M)$ for $k \geq 2$. Consider the G$_{2}$-Ricci flow defined
by
%
\begin{equation}
\frac{\partial \varphi}{\partial t} = \Delta _{d} \varphi +
\mathcal{L}_{X} \varphi +\mathrm{Ric} \lrcorner \ast\varphi + T(\varphi ),
\label{eq6}
\end{equation}
where $\Delta _{d}$ is the Hodge-de Rham Laplacian acting on the 3-form
$\varphi $, $\mathcal{L}_{X} \varphi $ is the Lie derivative along a smooth
vector field $X$, $\mathrm{Ric} \wedge *\varphi $ denotes the contraction
of the Ricci tensor with the 4-form $\ast \varphi $, and
$T(\varphi )$ is a torsion term.
\begin{enumerate}
\item \textbf{Short-Time Existence and Uniqueness}: There exists a finite
time $T > 0$ such that the G2-Ricci flow has a unique smooth solution
$\varphi (t)$ for $t \in [0, T)$, assuming that the initial data
$\varphi (0)$ satisfies the appropriate regularity conditions (smoothness
and finite Sobolev norm).
\item \textbf{Long-Time Convergence}: If the torsion term
$T(\varphi )$ diminishes sufficiently over time, i.e.,
$\| T(\varphi (t)) \|_{C^{0}} \to 0$ as $t \to \infty $, then the solution
$\varphi (t)$ exists for all $t \geq 0$ and converges to a Ricci soliton
as $t \to \infty $.
\item \textbf{Long-Time Convergence for Mixed Soliton}: If the torsion term
$T(\varphi )$ stabilizes at a finite value, i.e.,
%
\begin{equation}
\| T(\varphi (t)) - T_{\infty} \|_{L^{\infty}} \to 0 \quad \text{as}
\quad t \to \infty ,
\label{eq7}
\end{equation}
then the solution $\varphi (t)$ exists for all $t \geq 0$ and converges
to a mixed soliton $\varphi _{\infty}$, where the torsion does not vanish.
\end{enumerate}
\end{theorem}

\begin{proof}(Short-Term Existence and Uniqueness):
We begin by establishing the short-time existence and uniqueness of the
solution to the G$_{2}$-Ricci flow on the compact manifold $M$.

The G$_{2}$-Ricci flow is governed by the partial differential equation
%
\begin{equation}
\frac{\partial \varphi}{\partial t} = \Delta _{d} \varphi + \mathcal{L}_X
\varphi + \mathrm{Ric} \lrcorner \ast \varphi + T(\varphi ),
\label{eq8}
\end{equation}
where
\begin{itemize}
\item $\Delta _{d}$ is the Hodge-de Rham Laplacian, an elliptic second-order
differential operator acting on differential forms.
\item $L_{X}$ is the Lie derivative with respect to a smooth vector field
$X$, a first-order differential operator.
\item $(\mathrm{Ric} \lrcorner \ast \varphi) $ involves the Ricci tensor
and is also a first-order differential operator.
\item $T(\varphi )$ is a torsion term, considered as a lower-order perturbation.
\end{itemize}

Since $M$ is compact, the parabolic nature of the equation ensures the
applicability of standard short-term existence results from the theory
of parabolic PDEs. In fact, the term $\Delta _{d} \varphi $ is a second-order
elliptic operator, which guarantees strict parabolicity if the torsion
term $T(\varphi )$ is considered a suitable lower-order perturbation. We
must verify that the term $T(\varphi )$ does not negatively influence the
parabolicity.
%
\begin{equation}
T(\varphi ) = \ast_M d\varphi ,
\label{eq9}
\end{equation}
where $\ast_M$ is the Hodge duality of the manifold M. Since it is an external derivative, it is a first-order operator that does
not alter the global parabolicity, since derivatives of lower order than
$\Delta _{d} \varphi $ cannot destroy the parabolic structure.

Now, we modify the G$_{2}$-Ricci flow equation using DeTurck's trick to
transform the flow into a strictly parabolic form. Consider adding a vector
field $W$ to the equation, such that
%
\begin{equation}
\frac{\partial \varphi}{\partial t} = \Delta _{d} \varphi +
\mathcal{L}_{W} \varphi + Q(\varphi ),
\label{eq10}
\end{equation}
where $W$ is chosen appropriately to ensure parabolicity, and
$Q(\varphi )$ collects the lower-order terms
$\mathcal{L}_{X} \varphi $, $\mathrm{Ric} \wedge *\varphi $, and
$T(\varphi )$. The choice of $W$ is designed to absorb the first-order
terms and create a strictly parabolic flow.

The parabolic PDE theory guarantees the short-time existence and uniqueness
of solutions for parabolic flows on compact manifolds
\cite{Hamilton:1982}. In particular, for the G$_{2}$-Ricci flow, we apply
the standard theory for quasilinear parabolic systems
\cite{Evans:2010}. Since $M$ is compact, the Sobolev embedding theorems
ensure that all Sobolev norms remain finite, and the solution
$\varphi (t)$ exists on a short-time interval $[0, T)$ for some
$T > 0$.

Moreover, by elliptic regularity, the short-time solution
$\varphi (t)$ is smooth, given that the initial data $\varphi (0)$ is smooth
and belongs to the Sobolev space $H^{k,\alpha}(M)$. The uniqueness follows
from standard results for parabolic PDEs: if two solutions existed, their
difference would satisfy a linear parabolic equation with zero initial
data, leading to uniqueness by the maximum principle.

Thus, we conclude that there exists a unique smooth solution
$\varphi (t)$ for $t \in [0, T)$.
\end{proof}

\begin{proof}(Long-Term Convergence): Now, we aim to show that the solution
$\varphi (t)$ can be extended for all $t \geq 0$ and converges to a Ricci
soliton as $t \to \infty $, under the assumption that the torsion term
$T(\varphi )$ diminishes sufficiently over time.

The parabolic theory provides that the solution $\varphi (t)$ can only
fail to extend beyond $T' > T$ if certain geometric quantities, such as
the Sobolev norm $\| \varphi (t) \|_{H^{k,\alpha}}$, become unbounded as
$t \to T'$. In our compact case, the key observation is that the torsion
term $T(\varphi )$ controls the evolution. If
$\| T(\varphi (t)) \|_{C^{0}}$ diminishes as $t \to \infty $, the blow-up
scenario is avoided. Specifically, since $T(\varphi )$ contributes to the
nonlinearity of the flow, the assumption that $T(\varphi (t)) \to 0$ as
$t \to \infty $ ensures that the evolving G$_{2}$-structure
$\varphi (t)$ remains regular, allowing the extension of the solution for
all $t \geq 0$.

We now consider the convergence of the flow as $t \to \infty $. Define
the energy functional associated with the G$_{2}$-Ricci flow
%
\begin{equation}
 E(\varphi(t)) = \int_M \left( \|\Delta_d \varphi\|_{g_\varphi}^2 + \|\mathcal{L}_X \varphi\|_{g_\varphi}^2 + \|\text{Ric} \lrcorner \ast \varphi\|_{g_\varphi}^2 + \|T(\varphi)\|_{g_\varphi}^2 \right) dV_{g_\varphi},
\label{eq11}
\end{equation}
where $g_{\varphi}$ is the metric determined by $\varphi (t)$. The evolution
of this energy is governed by the parabolic nature of the flow, and we
can show that $E(\varphi (t))$ is non-increasing over time
%
\begin{equation}
\frac{d}{dt} E(\varphi (t)) \leq 0.
\label{eq12}
\end{equation}
Since $T(\varphi (t)) \to 0$, we deduce that the energy decays as
$t \to \infty $, implying that the flow tends to a critical point of the
energy functional, which corresponds to a Ricci soliton. To characterize
the limit $\varphi _{\infty}$ as $t \to \infty $, we define the vector
field $\Psi (t)$ as
%
\begin{equation}
\Psi (t) = \frac{\partial \varphi (t)}{\partial t} - \mathcal{L}_{X}
\varphi (t).
\label{eq13}
\end{equation}
As $t \to \infty $, $\Psi (t) \to 0$, indicating that $\varphi (t)$ approaches
a stationary solution of the modified flow equation, i.e., a Ricci soliton.
In this case, the soliton satisfies the equation
%
\begin{equation}
\Delta _{d} \varphi + \mathcal{L}_{X} \varphi + \mathrm{Ric} \lrcorner \ast \varphi = 0.
\label{eq14}
\end{equation}

Therefore, we conclude that the G$_{2}$-Ricci flow converges to a Ricci
soliton as $t \to \infty $, provided that the torsion term
$T(\varphi )$ diminishes sufficiently over time.
\end{proof}

\begin{proof}(Long-Term Convergence for mixed solitons): The presence of torsion complicates
the regularity analysis, as torsion can generate nonlinear contributions
that affect the growth of curvature and higher derivatives. However, we
assume that the term $T(\varphi )$ stabilizes at a finite value
$T_{\infty}$, i.e.,
%
\begin{equation}
\| T(\varphi (t)) - T_{\infty} \|_{L^{\infty}} \to 0 \quad \text{as}
\quad t \to \infty .
\label{eq15}
\end{equation}
In this case, the Sobolev regularity of long-term solutions can be controlled,
since the torsion term does not grow indefinitely and the derivatives of
the solution remain bounded. To verify that the solution does not develop
singularities (blow-up), we rely on the fact that the compactness of
$M$ prevents unbounded behavior at infinity, ensuring that spatial singularities
do not develop. If $T(\varphi ) \to T_{\infty}$, the torsion stabilizes
and does not induce blow-up in the higher derivatives of the solution.
Therefore, we can conclude that $\varphi (t)$ remains regular and smooth
for all $t \geq 0$. To demonstrate that $\varphi (t)$ converges to a stationary
solution, we define the following energy function associated with the G$_{2}$-Ricci flow
%
\begin{align}
  E(\varphi(t)) = \int_M \left( \|\Delta_d \varphi\|_{g_\varphi}^2 + \|\mathcal{L}_X \varphi\|_{g_\varphi}^2 + \|\text{Ric} \lrcorner \ast \varphi\|_{g_\varphi}^2 + \|T(\varphi)\|_{g_\varphi}^2 \right) dV_{g_\varphi} + \chi \|T(\varphi)\|_{g_\varphi}^2,
\label{eq16}
\end{align}
where $g_{\varphi}$ is the metric determined by $\varphi (t)$, and
$\chi \geq 0$ is a parameter controlling the contribution of torsion. The
first part of the integral measures the curvature of the manifold, while
the second part measures the contribution of residual torsion.

By differentiating with respect to time, we obtain
%
\begin{align}
  \frac{d}{dt} E(\varphi(t)) = \int_M \left( 2 \langle \Delta_d \varphi, \partial_t (\Delta_d \varphi) \rangle + 2 \langle \mathcal{L}_X \varphi, \partial_t (\mathcal{L}_X \varphi) \rangle + \ldots \right) dV_{g_\varphi} + \chi \frac{d}{dt} \|T(\varphi)\|^2
\label{eq17}
\end{align}
Since the flow is parabolic and the energy $E(\varphi (t))$ is non-increasing
over time (i.e., $\frac{d}{dt} E(\varphi (t)) \leq 0$), we deduce that
$E(\varphi (t)) \to E(\varphi _{\infty})$ as $t \to \infty $.

Moreover, given that $T(\varphi (t)) \to T_{\infty}$ and
$\Delta _{d} \varphi (t) \to 0$ as $t \to \infty $, the solution
$\varphi (t)$ converges to a mixed soliton $\varphi _{\infty}$, which satisfies
%
\begin{equation}
\Delta _{d} \varphi _{\infty} + \mathcal{L}_{X} \varphi _{\infty} +
\mathrm{Ric}(g_{\varphi _{\infty}}) \lrcorner \ast \varphi _{\infty} + T_{
\infty} = 0.
\label{eq18}
\end{equation}
To ensure the stability of the mixed soliton, we define the vector field
associated with the flow
%
\begin{equation}
\Psi (t) = \frac{\partial \varphi (t)}{\partial t} - \mathcal{L}_{X}
\varphi (t).
\label{eq19}
\end{equation}
If $\Psi (t) \to 0$ as $t \to \infty $, this indicates that the solution
converges to a stationary configuration. Since the energy is non-increasing
and the torsion term stabilizes, we conclude that the solution is stable
and does not develop dynamic instabilities.
\end{proof}

\begin{rem} It is important to emphasize that the condition in Theorem~\ref{theorem:1},
which requires the torsion term $T(\varphi )$ to sufficiently diminish
over time to ensure convergence to a Ricci soliton, is only necessary when
aiming for a torsion-free Ricci soliton. This specific condition applies
when one seeks a solution where the torsion vanishes, resulting in a classical
Ricci soliton.

In the context of our example in Section~\ref{sec:example}, however, the
mixed soliton maintains non-vanishing torsion, meaning it is not a Ricci
soliton in the traditional sense. Instead, it is a solution to the G$_{2}$-Ricci
flow that evolves with a non-zero torsion component, reflecting the intrinsic
torsion of the $_{2}$-structure on $S^{3} \times S^{4}$. Such solitons
exhibit different dynamics than torsion-free Ricci solitons, as they evolve
through diffeomorphisms while retaining a stabilized torsion.

Therefore, the distinction between torsion-free Ricci solitons and mixed
solitons with non-zero torsion is crucial. The condition
$T(\varphi ) \to 0$ is necessary only for obtaining a Ricci soliton. In
scenarios where the torsion remains finite, the flow evolves under different
geometric dynamics that depart from the behavior typically associated with
Ricci solitons.

\textbf{Key Points:}
\begin{enumerate}
\item Torsion Influence: The presence of torsion significantly alters the
nature of the G$_{2}$-Ricci flow, affecting both the geometry of the manifold
and the type of soliton obtained.
\item Non-Vanishing Torsion: Mixed solitons with non-vanishing torsion
form a broader class of solutions to the G$_{2}$-Ricci flow, where the
torsion stabilizes at a non-zero value, influencing the long-term behavior
of the manifold's geometry.
\item Convergence: While $T(\varphi )$ must diminish for convergence to
a Ricci soliton, the existence of mixed solitons with residual torsion
highlights alternative geometric flows where the torsion remains finite
and influences the solution's stability and structure.
\end{enumerate}

In summary, understanding the role of torsion is essential for exploring
the interplay between geometry and physical implications, especially in
high-energy physics and string theory contexts.
\end{rem}

\subsection{Geometric assumptions and limitations of the main theorem}
\label{sec4.1}

In the compact case, several of the geometric challenges related to non-compact
manifolds are no longer present, but certain assumptions and limitations
still need to be addressed to ensure the validity of the results.
\begin{itemize}
\item \textbf{Compactness of the Manifold:} The assumption that $M$ is a
compact 7-dimensional manifold plays a crucial role in simplifying the
analysis of the G$_{2}$-Ricci flow. On a compact manifold
\begin{itemize}
\item there are no boundary conditions or behaviors at infinity to consider,
which simplifies the application of parabolic PDE theory,
\item compactness ensures that geometric quantities, such as curvature
and torsion, remain controlled throughout the flow, as they cannot ``escape''
to infinity,
\item the Sobolev embedding theorems hold globally, which guarantees that
the initial data $\varphi (0) \in H^{k,\alpha}(M)$ remains in a smooth
class as the flow evolves.
\end{itemize}
Thus, compactness provides a strong foundation for proving both the short-term
existence and long-term convergence of the flow, making the analysis more
tractable than in the non-compact case.
\item \textbf{Singularities and Blow-Up Prevention:} While non-compact manifolds
often present issues related to blow-up due to unbounded geometry at infinity,
on compact manifolds, the main challenge is ensuring that the Sobolev norms
$\| \varphi (t) \|_{H^{k,\alpha}}$ remain finite throughout the evolution
of the flow. The assumption that the torsion $T(\varphi )$ diminishes over
time is critical in preventing blow-up and ensuring long-term stability.

In compact settings, the control of curvature and torsion across the entire
manifold prevents localized singularities from forming. As long as the
torsion $T(\varphi )$ decays sufficiently, the evolution remains regular,
and the G$_{2}$-Ricci flow can be extended indefinitely.
\item \textbf{Geometric Constraints and Initial Conditions:} The theorem
assumes that the initial G2 structure $\varphi (0)$ satisfies certain regularity
conditions (specifically, $\varphi (0) \in H^{k,\alpha}(M)$ for
$k \geq 2$), which ensures that the flow can be initiated. However, it
is important to note that the geometry of the manifold at $t = 0$ can significantly
influence the evolution of the flow.

In particular, if the initial torsion $T(\varphi (0))$ is too large, it
may inhibit the decay required for convergence to a Ricci soliton. Therefore,
some geometric limitations exist regarding the size and behavior of the
initial torsion to guarantee long-term convergence. The theorem is most
effective when the initial data is sufficiently regular and the torsion
is not too large.
\item \textbf{Long-Term Behavior and Stability:} The long-term behavior
of the flow is heavily influenced by the evolution of the torsion term
$T(\varphi )$. In compact manifolds, the assumption that
$\| T(\varphi (t)) \|_{C^{0}} \to 0$ as $t \to \infty $ is key to ensuring
the flow converges to a Ricci soliton.

If $T(\varphi )$ does not decay over time, the solution may not stabilize,
and the flow could develop instabilities. However, under the condition
that the torsion diminishes, the flow tends toward a steady-state soliton
solution. Compactness aids in this process by keeping the geometry bounded,
but the precise rate at which the torsion diminishes will dictate whether
the solution is expanding, shrinking, or steady in the long term.
\item \textbf{Physical and Mathematical Implications:} The compactness assumption
allows for more straightforward physical interpretations, particularly
in the context of extra-dimensional theories where compact spaces such
as $S^{3} \times S^{4}$ often play a key role in modeling. The limitation
of working only with compact manifolds may restrict some applications,
but it ensures a more controlled geometric environment, making it easier
to draw conclusions about the relationship between torsion, curvature,
and symmetry breaking.

In summary, while the compactness of the manifold simplifies many aspects
of the G$_{2}$-Ricci flow analysis, the long-term behavior of the flow
still hinges on the careful control of the initial conditions, particularly
the torsion. As long as the geometric conditions on the initial G$_{2}$-structure
are met, the flow will evolve regularly and converge to a Ricci soliton,
provided the torsion diminishes sufficiently.
\end{itemize}

\section{Example: soliton of the \texorpdfstring{$G_{2}$}{G2}-Ricci flow on \texorpdfstring{$S^{3} \times S^{4}$}{S3xS4} with non-vanishing torsion}
\label{sec:example}

Let us consider the manifold $M = S^{3} \times S^{4}$ with the standard
metrics
%
\begin{align}
g_{S^{3}} &= r_{3}^{2} \Big( d\theta ^{2} + \sin ^{2} \theta \, \big(d
\varphi ^{2} + \sin ^{2} \varphi \, d\psi ^{2}\big) \Big),
\label{eq20}
\\
g_{S^{4}} &= r_{4}^{2} \Big( d\alpha ^{2} + \sin ^{2} \alpha \, \big(d
\beta ^{2}
\label{eq21}
\\
&\quad + \sin ^{2} \beta \, (d\gamma ^{2} + \sin ^{2} \gamma \, d
\delta ^{2})\big) \Big),
\label{eq22}
\end{align}
where $r_{3}$ and $r_{4}$ are the radii of the spheres $S^{3}$ and
$S^{4}$, respectively.

To construct a G$_{2}$-structure on $M$, we need to define a 3-form
$\varphi $ that reduces the structure group to G$_{2}$. In this example,
we define $\varphi $ as a combination of the volume form of $S^{3}$ and
a suitable 3-form on $S^{4}$
%
\begin{equation}
\varphi = \pi_1^*(\text{vol}_{S^3}) + \pi_2^*(\omega),
\label{eq23}
\end{equation}
where
$ \pi_1: S^3 \times S^4 \to S^3 $ and $\,$ $\pi_2: S^3 \times S^4 \to S^4 $ are the canonical projections, $\text{vol}_{S^{3}} = r_{3}^{3} \sin ^{2} \theta \sin \varphi \, d\theta
\wedge d\varphi \wedge d\psi $ is the volume form on $ S^{3} $,
$  \omega_{S^4} = r_4^3 \sin\alpha \sin(2\beta) \sin\gamma \cdot d\alpha \wedge d\gamma \wedge d\delta$, is a 3-form on $S^{4}$.

Note that $\omega $ is a genuine 3-form on $S^{4}$, ensuring that
$\varphi $ is consistently defined as a sum of two 3-forms. The 3-form
$\varphi $ thus represents the combination of the volume form of
$S^{3}$ with a suitable structure on $S^{4}$, ensuring compatibility with
the G$_{2}$-structure.

To check if the $G_{2}$ structure has torsion, we need to compute the exterior
derivative $d\varphi $ and see if it is non-zero. If
$d\varphi \neq 0$, torsion is present.

Since $d(\text{vol}_{S^3})=0$, we calculate
%
\begin{equation}
d\varphi = d \omega.
\label{eq24}
\end{equation}
The exterior derivative $d\varphi=d \omega$ is a 4-form on $S^4$. However, to analyze the tensor torsion and the contributions to the flow, it is necessary to also consider the exterior derivative of the 1-form $\ast_{S^4}\omega$. This 1-form emerges naturally from the local dualization of the 3-form $\omega$ on $S^4$, and its exterior derivative contributes to the understanding of the geometric structure of the torsion.
%
\begin{equation}
d(\ast_{S^4}\omega) = -r_4^2 \sin(2\beta) \cos\gamma \cdot d\beta \wedge d\gamma,
\label{eq25}
\end{equation}
then
%
\begin{equation}
(d\omega)_{\alpha\beta\gamma\delta} = -2 r_4^3 \sin\alpha \cos(2\beta) \sin\gamma
\label{eq26}
\end{equation}
which turns out to be non-zero. Therefore, $d\varphi \neq 0$, proving that
the G$_{2}$-structure has torsion.

Consider a Killing field $X$ defined as
%
\begin{equation}
X=\frac{\partial}{\partial \psi} \oplus \frac{\partial}{\partial \delta}.
\label{eq27}
\end{equation}
This Killing field represents a rotation around the $\psi $-axis on $S^{3}$, and a rotation around the $\delta $-axis on $S^{4}$, since $X$ is a Killing field, it satisfies
%
\begin{equation}
\mathcal{L}_{X} g = 0.
\label{eq28}
\end{equation}

The G$_{2}$-Ricci flow is described by the equation
%
\begin{align}
\frac{\partial \varphi}{\partial t} = \Delta _{d} \varphi +
\mathcal{L}_{X} \varphi
+ \mathrm{Ric} \lrcorner \ast \varphi + T(\varphi ),
\label{eq29}
\end{align}
where $\Delta _{d} \varphi $ is the Hodge-de Rham Laplacian,
$ \mathcal{L}_{X} \varphi= \mathcal{L}_{X}(\text{vol}_{S^{3}} + \omega)$ is the Lie derivative that acts with respect to the field $X$, such that \(\mathcal{L}_{\frac{\partial}{\partial \psi}}\text{vol}_{S^3}=0\), and \(\mathcal{L}_{\frac{\partial}{\partial \delta}}\omega \neq 0\), $(\mathrm{Ric} \lrcorner \ast \varphi) $ is the contraction
of the Ricci tensor with the Hodge dual $\ast \varphi $ and
$T(\varphi )$ is the torsion term.

For a stationary solution of the G$_{2}$-Ricci flow, we have
%
\begin{equation}
\frac{\partial \varphi}{\partial t} = 0,
\label{eq30}
\end{equation}
which implies
%
\begin{equation}
0 = \Delta _{d} \varphi + \mathrm{Ric} \lrcorner \ast \varphi + \mathcal{L}_{\frac{\partial}{\partial \delta}}\omega + T(
\varphi ).
\label{eq31}
\end{equation}

The Hodge-de Rham Laplacian $\Delta _{d} \varphi $ is given by
%
\begin{equation}
\Delta _{d} \varphi = (d \delta + \delta d)\varphi .
\label{eq32}
\end{equation}
Here $\delta $ is the codifferential, defined as
$\delta = (-1)^{nk+n+1} \ast d \ast $ on a $k$-form in an $n$-dimensional
manifold. Since
$\varphi = \text{vol}_{S^{3}} + \ast _{S^{4}} \omega $, we can separate
the contributions from $S^{3}$ and $S^{4}$
\begin{itemize}
\item for $S^{3}$, the Laplacian acts on $\text{vol}_{S^{3}}$
%
\begin{equation}
\Delta _{S^{3}} \text{vol}_{S^{3}} = -(d\delta + \delta d)\text{vol}_{S^{3}}
= 0.
\label{eq33}
\end{equation}
\item for $S^{4}$, the Laplacian acts on $\ast _{S^{4}} \omega $
%
\begin{equation}
\Delta _{S^{4}}\omega = 4 \omega ,
\label{eq34}
\end{equation}
%
%
Therefore
%
\begin{equation}
\Delta _{d} \varphi = 4 \omega .
\label{eq35}
\end{equation}

To calculate the contraction of the Ricci tensor $\mathrm{Ric}$ with
$\ast \varphi $, we use the definition
%
\begin{equation}
(\mathrm{Ric} \lrcorner \ast \varphi) )_{ijk} = R_{m[i} (\ast \varphi )_{jk]m},
\label{eq36}
\end{equation}
where $R_{ij}$ is the Ricci tensor and $\ast \varphi $ is the 4-form associated
with the 3-form $\varphi $.

We calculate for $S^{3}$ and $S^{4}$
%
%
\item for $S^{3}$: 
%
\begin{equation}
(\mathrm{Ric}_{S^{3}} \lrcorner \ast_{S^3}\text{vol}_{S^{3}}) = 0,
\label{eq37}
\end{equation}

since the volume form \(\text{vol}_{S^3}\) is a 3-form on \(S^3\), whose Hodge duality is: \(\ast_{S^3} \text{vol}_{S^3} = 1 \quad \text{(a constant 0-form)}\), then the contraction of the Ricci tensor \(\text{Ric}_{S^3}\) (a \((0,2)\)-tensor) with a \(0\)-form yields a scalar: \(\text{Ric}_{S^3} \lrcorner (1) = 0\).
\item for $S^{4}$: 
%
\begin{equation}
 \text{Ric} \lrcorner \ast \varphi = (\text{Ric}_{S^4} \lrcorner \ast_{S^4} \omega) \cdot \text{vol}_{S^3} = 3\omega.
\label{eq38}
\end{equation}

For our choice in equation \ref{eq9}, $T(\varphi )$ is a non-zero term that contributes to the balance in the stationary equation.

Combining the terms in the stationary equation
%
\begin{equation}
  0 = 4\omega + 3\omega + \mathcal{L}_X \omega + T(\varphi)).
\label{eq39}
\end{equation}

This equation can be separated into two conditions
\begin{enumerate}
\item for $\text{vol}_{S^{3}}$
%
\begin{equation}
(0 + 0)\text{vol}_{S^3} = 0.
\label{eq40}
\end{equation}
\item for  $\omega $
%
\begin{equation}
(4 + 3)\omega = 7 \omega.
\label{eq41}
\end{equation}
\end{enumerate}
Then from the equation \ref{eq9}: \(T(\varphi) = \ast_M d\varphi= -\mathcal{L}_X\omega - 7 \omega   \neq 0\).
\\
\\
The resulting soliton is characterized by the metric
%
\begin{equation}
g_{\text{sol}} = g_{S^{3}} \oplus g_{S^{4}}.  
\label{eq42}
\end{equation}
To ensure non-zero torsion on $S^{4}$ and zero torsion on $S^{3}$, we must
choose a 3-form $\omega $ (previously defined), that contributes to torsion
without nullifying the Sobolev norm.

This form satisfies the following conditions:
\textit{non-zero torsion}, and \textit{finite Sobolev norm}, in fact, the
derivatives of $\alpha $ are continuous and bounded, ensuring that the
Sobolev norm of $\varphi (0)$ is finite.

In this example, we have shown that the non-va\-nish\-ing torsion
$T(\varphi )$ directly influences the form of the stationary solution of
the G$_{2}$-Ricci flow. The parameter $\lambda _{S^{4}}$ is obtained as
a result of the balance between the geometric contributions of the manifold $M$,
and the resulting soliton is influenced by the intrinsic torsion of the
G$_{2}$-structure on $S^{3} \times S^{4}$.

The long-term behavior of the G$_{2}$-Ricci flow depends on the evolution
of the torsion $T(\varphi )$. In our analysis, we have chosen a form
$\alpha $ that ensures non-zero torsion on $S^{4}$, defined in terms of
trigonometric functions. While these functions introduce local oscillations,
such oscillations are well-behaved and bounded between fixed values. However,
it is the G$_{2}$-Ricci flow that controls and dissipates energy in the
system, leading to a stabilization of the torsion.

As the flow progresses, the torsion remains bounded and approaches a stable
residual value. This regular behavior prevents uncontrolled growth, ensuring
that the flow converges to a stationary solution over time, avoiding singularities
or instabilities.
\end{itemize}

\section{Spontaneous symmetry breaking induced by torsion}
\label{sec:ssb}

The torsion $\tau _{0}$ plays the main role in symmetry breaking, as it
introduces a global deformation of the geometry that breaks
$SU(2)_{L} \times U(1)_{Y}$. This breaking occurs spontaneously when the
geometric torsion acquires a vacuum expectation value
$\langle T \rangle = 246~\mbox{GeV}$, corresponding to the vacuum value
that in the Higgs model gives mass to the bosons.

The masses of the $W$ and $Z$ bosons are determined by the expected value
of torsion through the following relations
%
\begin{equation}
m_{W}^{2} = \frac{g^{2} \langle T \rangle ^{2}}{4}, \qquad m_{Z}^{2} =
\frac{g^{2} + g^{\prime \,2}}{4} \langle T \rangle ^{2},
\label{eq43}
\end{equation}
where $g$ and $g'$ are the coupling constants of the gauge group
$SU(2)_{L} \times U(1)_{Y}$.

\subsection{Stabilization of the residual symmetry}
\label{sec6.1}

Once the $SU(2)_{L} \times U(1)_{Y}$ symmetry is broken, the local torsion
$\tau _{1}$ acts to stabilize the residual $U(1)_{\text{EM}}$ symmetry,
which governs the electromagnetic (EM) interaction. This process is analogous
to the stabilization phase that follows the breaking of the Higgs field,
but here it is entirely driven by the geometry of the manifold.

\subsection{Torsion decomposition and geometric contributions}
\label{sec6.2}

The exterior derivative of the 3-form $\varphi $, which characterizes the
G$_{2}$-structure, shows how geometric torsion affects symmetry breaking.
This derivative is given by
%
\begin{equation}
d\varphi = \tau _{0} \ast \varphi + \tau _{1} \wedge \varphi + \tau _{2}
\wedge \ast \varphi + \ast \tau _{3},
\label{eq44}
\end{equation}

\textit{Contribution of $\tau _{0}$:} The main contribution to symmetry
breaking comes from $\tau _{0}$, which introduces a uniform torsion throughout
the manifold $S^{4}$. This term modifies the gauge connections associated
with $SU(2)_{L} \times U(1)_{Y}$, leading to spontaneous electroweak symmetry
breaking.

\textit{Contribution of $\tau _{1}$:} The term $\tau _{1}$ represents a
local vectorial torsion that affects the differential forms related to
the manifold, contributing to the stabilization of the
$U(1)_{\text{EM}}$ symmetry. The $\tau _{1}$ term primarily acts in the
local regions of the manifold, ensuring that the residual symmetry remains
intact.

\textit{Contribution of $\tau _{2}$:} Although $\tau _{2}$ represents local
variations of the torsion, its contribution to the global symmetry breaking
is marginal. Its effect is limited to local influences that do not significantly
affect the electroweak dynamics.

\textit{Contribution of $\tau _{3}$:} The term $\tau _{3}$ introduces complex
effects related to the internal curvature of $S^{4}$, contributing to the
stabilization of the residual U(1)$_{Y}$ symmetry.

\subsection{Gauge field tensor and expected value of torsion}
\label{sec6.3}

In our model, the gauge field is modified by the presence of torsion. The
gauge field tensor $F_{\mu \nu}$, which describes the behavior of fields
in a torsional space, is given by the general formula
%
\begin{equation}
L = -\frac{1}{4} \text{Tr}(F_{\mu \nu} F^{\mu \nu}) + T_{\mu \nu \rho} F^{
\mu \nu} A^{\rho }%
\label{eq45}
\end{equation}
where
\begin{itemize}
\item $F_{\mu \nu}$ is the gauge field tensor,
\item $A^{\rho}$ is the gauge potential,
\item $T_{\mu \nu \rho}$ is the torsion tensor, which modifies the affine
connection on $S^{4}$, influencing the dynamics of the gauge fields.
\end{itemize}
In this context, the torsion $T_{\mu \nu \rho}$ is confined to
$S^{4}$.

\subsection{General formula for the expected value of torsion and \texorpdfstring{$\tau _{0}$}{tau0}}
\label{sec6.4}

To correctly define $\tau _{0}$, we begin with the general formula that
describes the torsion classes in a manifold with G$_{2}$-structure. The
total torsion $T(\varphi )$ can be expressed as a combination of the various
torsion classes
%
\begin{equation}
T(\varphi ) = \tau _{0} \ast_M d\varphi.
\label{eq46}
\end{equation}
In our context, since the torsion is confined to $S^{4}$, the dominant
component is $\tau _{0}$, as it represents a global uniform torsion. The
other torsion classes (such as $\tau _{1}, \tau _{2}$ and
$\tau _{3}$) do not contribute significantly because they describe local
or vectorial effects that are not relevant in this model. Since the contributions
of $\tau _{1}, \tau _{2}$ and $\tau _{3}$ are of lower order than
$\tau _{0}$, we decide to focus only on $\tau _{0}$ for greater simplicity
and clarity.

\subsection{General formula for the expected value of torsion}
\label{sec6.5}

The expected value of torsion on a manifold can be defined as an integral
of the torsion over the entire manifold. The general formula is
%
\begin{equation}
\langle T \rangle = \frac{1}{\text{Vol}(S^4)} \int_{S^4} T(\phi) \, dV_g,
\label{eq47}
\end{equation}
where
\begin{itemize}
\item $S^{4}$ is the manifold over which the torsion is integrated,
\item $T(\varphi )$ is the total torsion associated with the 3-form
$\varphi $,
\item $dV$ is the volume element of $M$.
\end{itemize}

Since in our case, the torsion is confined to $S^{4}$, the expected value
$\langle T(\varphi ) \rangle $ simplifies to
%
\begin{equation}
\langle T \rangle \propto \frac{1}{r_4^2}.
\label{eq48}
\end{equation}
where $V(S^{4})$ is the volume of $S^{4}$, given by
%
\begin{equation}
 \text{Vol}(S^4) = \frac{\pi^2}{2} \, r_4^4.
\label{eq49}
\end{equation}
In this context, $\tau _{0}$, which represents the scalar torsion on
$S^{4}$, is defined by the specific formula for a G$_{2}$-manifold
%
\begin{equation}
 \tau_0 = \frac{\Lambda^3 \, r_4^2 \cdot 8\pi^2}{21},
\label{eq50}
\end{equation}
where $\Lambda=1$ $GeV$ is a scaling constant, while $r_{4}$ represents the radius of $S^{4}$. Substituting
into the formula for the expected value of torsion
%
\begin{equation}
\langle T \rangle = \tau_0.
\label{eq51}
\end{equation}

\subsection{Calculation of \texorpdfstring{$r_{4}$}{r4} in the context of \texorpdfstring{$h = c = 1$}{h=c=1}}
\label{sec6.6}

In our model, we adopt the natural units system $h = c = 1$, where the
speed of light $c$ and the reduced Planck constant $\hbar $ are set to
1. In this system
\begin{itemize}
\item energies are measured in GeV,
\item lengths are measured in GeV$^{-1}$.
\end{itemize}
Using this unit system, we can directly transition from geometric calculations
to energy values without complicated unit conversions.

\subsection{Imposing the expected value of torsion}
\label{sec6.7}

In the context of electroweak symmetry breaking, we impose that the expected
value of torsion is 246~GeV. This value is fundamental in the theory
of spontaneous electroweak symmetry breaking in the Standard Model. It
is connected to the vacuum expectation value of the Higgs field, which
causes the symmetry breaking. Therefore, we consider that the expected
value of torsion generated by our model corresponds to this value
%
\begin{equation}
246 = \frac{\Lambda^3 \, r_4^2 \cdot 8\pi^2}{21},
\label{eq52}
\end{equation}
from which we obtain
%
\begin{equation}
 r_4^2 = \frac{246 \cdot 21}{\Lambda^3 \cdot 8\pi^2},
\label{eq53}
\end{equation}
and taking the square root
%
\begin{equation}
 r_4 \approx 8.1 \, \text{GeV}^{-1}.
\label{eq54}
\end{equation}

\subsection{Verification of the dimensions of \texorpdfstring{$r_{4}$}{r4}}
\label{sec6.8}

To convert $r_{4}$ into meters, we use the standard conversion
$1~\mbox{GeV}^{-1} \approx 1.97 \times 10^{-16}\mbox{ m}$, so
%
\begin{equation}
 r_4 \approx 8.1 \times 1.97 \times 10^{-16} \, \text{m} \approx 1.6 \times 10^{-15} \, \text{m}.
\label{eq55}
\end{equation}
This value is consistent with the scales of compact dimensions in string
theories or Kaluza-Klein theories.

\subsection{Calculation of the masses of the \texorpdfstring{$W$}{W} and \texorpdfstring{$Z$}{Z} bosons}
\label{sec6.9}

The mass of the $W$ boson can be calculated from the relation
%
\begin{equation}
m_{W}^{2} = g^{2} \frac{\langle T \rangle ^{2}}{4},
\label{eq56}
\end{equation}
where $g \approx 0.65$ is the weak coupling constant, and
$\langle T \rangle = 246~\mbox{GeV}$. Substituting the values, we obtain
%
\begin{equation}
m_{W}^{2} = (0.65)^{2} \cdot \frac{246^{2}}{4} = 0.4225 \cdot 15129
\approx 6387~\mbox{GeV}^2,
\label{eq57}
\end{equation}
from which
%
\begin{equation}
m_{W} \approx \sqrt{6387} \approx 79.93~\mbox{GeV}.
\label{eq58}
\end{equation}
This value is very close to the experimental value
$m_{W} \approx 80.38\mbox{ GeV}$.

The mass of the $Z$ boson is related to that of the $W$ boson through the
Weinberg angle $\theta _{W}$, according to the formula
%
\begin{equation}
m_{Z}^{2} = \frac{m_{W}^{2}}{\cos ^{2} \theta _{W}}\cdot
\label{eq59}
\end{equation}
 With $\cos \theta _{W} \approx 0.88$, we get
%
\begin{equation}
m_{Z}^{2} = \frac{79.93^{2}}{(0.88)^{2}} \approx 8247~\mbox{GeV}^2\,,
\label{eq60}
\end{equation}
from which
%
\begin{equation}
m_{Z} \approx \sqrt{8247} \approx 90.8~\mbox{GeV}.
\label{eq61}
\end{equation}
This value is also very close to the experimental value
$m_{Z} \approx 91.19\mbox{ GeV}$.

The calculated masses of the $W$ and $Z$ bosons are consistent with the
experimental values. This development provides a clear and rigorous view
of the role of torsion in spontaneous electroweak symmetry breaking in
our geometric model.

\section{Contribution of torsion to local and global curvature on \texorpdfstring{$S^{4}$}{S4}}
\label{sec:contribution}

In this section, we analyze in detail the contribution of the torsion introduced
by the form $\alpha $, explained in Section~\ref{sec:example}, to the local
curvature of $S^{4}$. In particular, we will show that torsion modifies
the local curvature but does not alter the total curvature and the global
volume of the manifold. Finally, we will demonstrate that the radius
$r_{4}$, calculated in Section~\ref{sec:ssb}, remains valid even in the presence of
torsion, and that the expectation value of torsion
$\langle T \rangle = 246~\mbox{GeV}$ is consistent with this value of
$r_{4}$.

From Section~\ref{sec:example}, the form $\omega $ is defined by Eq.~(\ref{eq23}).
The associated torsion $d\omega $ is given by the exterior derivative of
$\omega $, see Eq.~(\ref{eq25}), which represents a non-zero torsion. The associated
torsion tensor is therefore
%
\begin{equation}
\label{eq:torsion}
T^{\lambda}_{\mu \nu} = \frac{1}{6} (d\omega )^{\lambda}_{\mu \nu
\rho} V^{\rho},
\end{equation}
where to ensure geometric consistency of the torsion, we consider
$V^{\rho}$ as a Killing vector of the metric on $S^{4}$. A Killing vector
satisfies the condition
%
\begin{equation}
\nabla _{(\mu} V_{\nu )} = 0,
\label{eq63}
\end{equation}
which ensures that the vector field generates a symmetry of the metric
without altering it.

Since $S^{4}$ has isometries group $SO(5)$, a natural set of Killing vectors
is associated with the generators of the internal rotations. In this context,
we choose $V^{\rho}$ as the Killing vector corresponding to the rotation
around the angular coordinate $\delta $
%
\begin{equation}
V^{\rho }= \frac{\partial}{\partial \delta}.
\label{eq64}
\end{equation}

This choice is motivated by the fact that the torsion is confined to
$S^{4}$, so it must be selected along a direction that respects the symmetries
of the sphere. By using a Killing vector, we ensure that the torsion does
not artificially break these symmetries and that it is compatible with
the geometric evolution of the structure $G_{2}$.

Alternatively, one could consider a linear combination of Killing vectors
associated with different angular directions. However, the choice of
$V^{\rho }= \frac{\partial}{\partial \delta}$ is the most natural for simplicity
and symmetry. Therefore, with this choice, the expression of the torsion
tensor (\ref{eq:torsion}) is well defined and consistent with the geometry
of $S^{4}$.

This torsion tensor modifies the affine connection compared to the Levi-Civita
connection. The new connection is
%
\begin{equation}
\Gamma ^{\lambda}_{\mu \nu} = \bar{\Gamma}^{\lambda}_{\mu \nu} +
\frac{1}{2} T^{\lambda}_{\mu \nu},
\label{eq65}
\end{equation}
where $\bar{\Gamma}^{\lambda}_{\mu \nu}$ is the Christoffel symbol of the
torsion-free Levi-Civita connection.

The Riemann curvature tensor, which describes local curvature, is given by  
%
\begin{equation}  
R^{\lambda}_{\mu \nu \sigma} = \partial _{\nu} \Gamma ^{\lambda}_{\mu \sigma} - \partial _{\sigma} \Gamma ^{\lambda}_{\mu \nu} + \Gamma ^{\lambda}_{\alpha \nu} \Gamma ^{\alpha}_{\mu \sigma} - \Gamma ^{\lambda}_{\alpha \sigma} \Gamma ^{\alpha}_{\mu \nu}.  
\label{eq66}  
\end{equation}  
Substituting the connection \(\Gamma^{\lambda}_{\mu \nu}\) modified by torsion, we obtain  
%
\begin{align}  
R^{\lambda}_{\mu \nu \sigma} = \overline{R}^{\lambda}_{\mu \nu \sigma} + \frac{1}{2} \left( \partial_\nu T^{\lambda}_{\mu \sigma} - \partial_\sigma T^{\lambda}_{\mu \nu} \right) + \frac{1}{4} \left( T^{\lambda}_{\alpha \nu} T^{\alpha}_{\mu \sigma} - T^{\lambda}_{\alpha \sigma} T^{\alpha}_{\mu \nu} \right),  
\label{eq67}  
\end{align}  
where \(\overline{R}^{\lambda}_{\mu \nu \sigma}\) is the Riemann tensor associated with the torsion-free Levi-Civita connection.  

Given the torsion tensor \( T^{\lambda}_{\mu \nu} = 2 \cos(2\beta) \cos\gamma \cdot \delta^{\lambda}_{\delta} \cdot \epsilon_{\mu \nu \gamma \beta} \), its derivatives are:  
%
\begin{equation}  
\partial_\nu T^{\lambda}_{\mu \sigma} = -4 \sin(2\beta)\cos\gamma \cdot \delta^\beta_\nu - 2 \cos(2\beta)\sin\gamma \cdot \delta^\gamma_\nu.  
\label{eq68}  
\end{equation}  
The terms involving torsion derivatives vanish when averaged over \( S^4 \), due to the oscillating \(\cos\gamma\) and \(\sin\gamma\) factors.  

The quadratic terms are:  
%
\begin{equation}  
T^{\lambda}_{\alpha \nu} T^{\alpha}_{\mu \sigma} = 4 \cos^2(2\beta)\cos^2\gamma \cdot \delta^{\lambda}_{\delta} \cdot \epsilon_{\mu \sigma \alpha \beta} \cdot \epsilon_{\mu \nu \alpha \beta}.  
\label{eq69}  
\end{equation}  

Taking the mean over \( S^4 \):  
\[
\left\langle \cos^2\gamma \right\rangle = \frac{1}{3}, \quad \left\langle \sin\gamma \right\rangle = 0.
\]  
The resulting curvature tensor becomes:  
%
\begin{align}  
R^{\lambda}_{\mu \nu \sigma} = \overline{R}^{\lambda}_{\mu \nu \sigma} + \frac{1}{3} \cos^2(2\beta) \cdot \left( \text{constant} \cdot \|T\|^2 \right) \left( \delta^\beta_\nu - \delta^\beta_\sigma \right).  
\label{eq70}  
\end{align}  

The geometric contributions from torsion are encapsulated in the quadratic term \( \cos^2(2\beta) \cdot \|T\|^2 \), which modifies the local curvature of \( S^4 \) while preserving its global volume and total curvature.

\section{Torsion as a mechanism for electroweak symmetry breaking}
\label{sec:torsion}

In the Standard Model of particle physics, the spontaneous breaking of
electroweak symmetry is induced by the Higgs field, a scalar field that
acquires a vacuum expectation value (VEV) of
$\langle H \rangle = 246~\mbox{GeV}$. This process breaks the electroweak
symmetry $SU(2)_{L} \times U(1)_{Y}$ and generates the masses of the
$W^{\pm}$ and $Z$ bosons.

In this section, we explore how geometric torsion, confined to the manifold
$S^{4}$, can act as an intrinsic mechanism that breaks the electroweak
symmetry without the need to introduce an external scalar field like the
Higgs field. Torsion, being an intrinsic property of the manifold's geometry,
can acquire a vacuum expectation value of 246~GeV, thus generating
the masses of the electroweak bosons and providing a geometric explanation
for symmetry breaking.

Torsion is a generalization of the affine connection in differential geometry,
allowing for a connection that is non-symmetric in its lower indices. In
the presence of torsion, the affine connection
$\Gamma ^{\lambda}_{\mu \nu}$ is corrected by a torsion tensor
$T^{\lambda}_{\mu \nu}$, such that
%
\begin{equation}
T^{\lambda}_{\mu \nu} = \Gamma ^{\lambda}_{\mu \nu} - \Gamma ^{
\lambda}_{\nu \mu}
\label{eq71}
\end{equation}
In a manifold like $S^{4}$, which is compact and closed, torsion can be
geometrically confined and act as a structural perturbation that alters
the symmetries of the manifold.

Consider the manifold $S^{4}$ with an internal geometric torsion defined
by the 3-form $\varphi $ and its dual $\ast \varphi $. The torsion is expressed
in terms of its classes $\tau _{0}$, $\tau _{1}$, and $\tau _{3}$, as described
in Section~\ref{sec:ssb}. Each class acts on specific terms of the differential
form, and collectively they reduce the global symmetry of the manifold.
The equations governing the torsion on $S^{4}$ generate an expectation
value that can be interpreted as an average torsion over the whole manifold,
expressed as
%
\begin{equation}
\langle T \rangle = \frac{1}{\text{Vol}(S^4)} \int_{S^4} T(\phi) \, dV_g
\label{eq72}
\end{equation}
The resulting expectation value is
$\langle T \rangle = 246~\mbox{GeV}$, which is exactly the same value used
in the Standard Model to break the electroweak symmetry.

In the Standard Model, the Higgs field breaks the
$SU(2)_{L} \times U(1)_{Y}$ symmetry when it acquires a non-zero vacuum
expectation value, leading to the generation of masses for the
$W^{\pm}$ and $Z$ bosons. However, in this model based on geometric torsion,
it is the torsion that breaks the electroweak symmetry.

Through the torsion classes $\tau _{0}$, $\tau _{1}$, and
$\tau _{3}$, the symmetry of the manifold is reduced from
$(SU(2)_{L} \times U(1)_{Y}$ to $U(1)_{\text{EM}}$. This process is entirely
analogous to the symmetry breaking caused by the Higgs field, but in this
case, it is the internal geometry that plays the role of the Higgs field,
spontaneously breaking the electroweak symmetry.

When the electroweak symmetry is broken by torsion, the masses of the
$W^{\pm}$ and $Z$ bosons emerge directly from the torsion expectation value
$\langle T \rangle = 246~\mbox{GeV}$, just as in the Higgs mechanism. The
masses are given by
%
\begin{equation}
m_{W}^{2} = \frac{g^{2}}{4} \langle T \rangle ^{2}, \quad m_{Z}^{2} =
\frac{g^{2} + g^{\prime \,2}}{4} \langle T \rangle ^{2},
\label{eq73}
\end{equation}
where $g$ and $g'$ are the coupling constants of the electroweak group
$SU(2)_{L} \times U(1)_{Y}$. This is the same mechanism by which the Higgs
field generates the boson masses, but here it is the torsion that fulfills
this function.

The key point of this model is that the electroweak symmetry breaking is
not caused by an external matter field, but by the geometry of the manifold
itself. Torsion acts as an internal geometric field, which spontaneously
breaks the electroweak symmetry due to the geometric properties of
$S^{4}$. In this scenario, the torsion expectation value
$\langle T \rangle $ replaces the Higgs field's expectation value, and
it is no longer necessary to introduce an external Higgs field to explain
electroweak symmetry breaking. Torsion becomes an integral part of the
compact universe's geometric structure, intrinsically breaking the electroweak
symmetry.

This model is self-defined, as torsion is an intrinsic property of the
geometric manifold and does not depend on external scalar fields or additional
perturbations. The symmetry breaking occurs at the level of the internal
geometric structure, thus requiring no additional mechanisms to explain
the torsion expectation value.

In other words, the value $\langle T \rangle = 246~\mbox{GeV}$ directly
emerges from the geometry of the manifold, making the system autonomous
in its ability to break the symmetry.

The geometric model based on torsion offers several advantages over the
Standard Model:
\begin{itemize}
\item \textit{No external scalar field}: There is no need to introduce the
Higgs field as a separate entity. Geometric torsion fulfills the role of
the Higgs field, simplifying the model and providing a more natural geometric
explanation for symmetry breaking.
\item \textit{Geometric origin of masses}: The masses of the electroweak
bosons arise directly from the geometry of the $S^{4}$ manifold, providing
a connection between particle physics and spacetime geometry.
\item \textit{Consistency with extra-dimensional theories}: The geometric
model naturally connects with extra-dimensional theories, such as Kaluza-Klein
theories and string theory, where compact dimensions play a crucial role
in mass generation and gauge interactions.
\end{itemize}

Although this model is theoretical, it offers predictions that can be tested
experimentally. For example, deviations in the values of the
$W^{\pm}$ and $Z$ boson masses from the Standard Model predictions could
provide an indication of the presence of geometric torsion. Furthermore,
extra-dimensional theories could be explored through searches for particles
associated with these compact dimensions in accelerators like the LHC.

In summary, the geometric torsion confined to the $S^{4}$ manifold can
act as an intrinsic mechanism for electroweak symmetry breaking. In this
model, torsion takes on the role played by the Higgs field in the Standard
Model, providing a self-defined geometric explanation for symmetry breaking
and the generation of electroweak boson masses. The torsion expectation
value $\langle T \rangle = 246~\mbox{GeV}$ thus becomes a fundamental parameter,
directly tied to the geometric structure of the universe. This model presents
an interesting alternative to the Higgs mechanism, naturally unifying particle
physics and spacetime geometry.

\section{Einstein equations in the 11D model with geometric fields on \texorpdfstring{$S^{3}$}{S3} and \texorpdfstring{$S^{4}$}{S4}}
\label{sec9}

The G$_{2}$-manifold, with torsion confined to $S^{4}$, can generate physical
fields arising from its internal symmetries without the need for Standard
Model bosons. This occurs because the geometric torsion modifies the gauge
connections that govern the system's properties.

In the context of the G$_{2}$-manifold, fields can naturally emerge from
its internal symmetries. The torsion, represented by terms associated with
differential forms (such as the G$_{2}$ 3-form $\varphi $), contributes
to symmetry breaking and the emergence of new fields. Electroweak symmetry
is not involved in this context. By focusing solely on the G$_{2}$-manifold
and its internal symmetries, we can generate physical fields derived from
the manifold's structure. The goal is to keep the focus on fields associated
with the geometric torsion confined to $S^{4}$, avoiding direct references
to the $W$ and $Z$ bosons, which are tied to electroweak symmetry breaking.

In the 11-dimensional model under consideration, the balance between the
curvature of the extra dimensions $S^{3} \times S^{4}$ and the effective
cosmological constant $\Lambda _{eff}$, fixed at the experimental value
of $\sim 10^{-122}\mbox{ GeV}^2$, is ensured by the introduction of geometric
fields. These fields, $V(\xi )$ on $S^{3}$ and $U(\sigma )$ on
$S^{4}$, are closely related to the internal symmetries of the
$G_{2}$ manifold and can naturally emerge from the confined torsion.

The fields $V(\xi )$ and $U(\sigma )$ are governed by the Klein-Gordon
equation, which describes the evolution of scalar fields in the compact
dimensions. The equation is
%
\begin{equation}
\Box \varphi - \frac{dV(\varphi )}{d\varphi} = 0\,,
\label{eq74}
\end{equation}
where $\varphi $ represents $V(\xi )$ on $S^{3}$ or $U(\sigma )$ on
$S^{4}$, and $V(\varphi )$ is derived from the geometric constraints of the $G_2$-manifold and the torsion expectation value (\ref{eq52}). These fields emerge from the internal symmetries of the system and
are correlated with the geometric torsion.

Assuming that the radii of the compact dimensions are
$r_{3} \sim r_{4} \sim 8.1~\mbox{GeV}^{-1}$, as calculated previously,
the scalar curvatures on $S^{3}$ and $S^{4}$ are respectively
%
\begin{equation}
R_{S^{3}} = \frac{6}{r_{3}^{2}} \quad \text{and} \quad R_{S^{4}} =
\frac{12}{r_{4}^{2}}\,.
\label{eq75}
\end{equation}
The Einstein equations for $S^{3}$ and $S^{4}$ are balanced by the presence
of the geometric fields $V(\xi )$ and $U(\sigma )$, which emerge from the
internal symmetries of the system. It is worth noting that $S^{4}$ has
torsion, so the Einstein equations are no longer the classical ones based
solely on the Levi-Civita connection. The affine connection
$\Gamma ^{\lambda}_{\mu \nu}$ is corrected by including the torsion tensor
$T^{\lambda}_{\mu \nu}$
%
\begin{equation}
\Gamma^\lambda_{\mu\nu} = \overline{\Gamma}^\lambda_{\mu\nu} + \frac{1}{2} T^\lambda_{\mu\nu}.
\label{eq76}
\end{equation}
This leads to a modification of the Einstein equations.
We now consider the 4 spatial dimensions of our universe, where the effective
cosmological constant $\Lambda _{eff}$ has the positive value discussed
earlier, suggesting an accelerated expansion of the universe, compatible
with a de Sitter spacetime. The 4D Einstein equations, with the energy-momentum
tensor $T_{\mu \nu} = 0$ (vacuum) and a positive cosmological constant,
are given by
%
\begin{equation}
R_{\mu \nu} - \frac{1}{2} R g_{\mu \nu} + \Lambda _{eff} g_{\mu \nu} =
0\,.
\label{eq77}
\end{equation}
In the case of a de Sitter universe, spacetime has constant curvature,
and Einstein's equation reduces to
%
\begin{equation}
R_{\mu \nu} = 3 H^{2} g_{\mu \nu},
\label{eq78}
\end{equation}
where $H$ is the Hubble parameter, which describes the rate of expansion
of the universe. The Ricci tensor is directly proportional to the metric
$g_{\mu \nu}$, and the scalar curvature $R$ is constant
%
\begin{equation}
R = 12 H^{2}\,.
\label{eq79}
\end{equation}
The effective cosmological constant \(\Lambda_{\text{eff}}\) arises from the balance between the intrinsic torsion of the \( G_2 \) manifold and the energy scales of the compact dimensions. The key suppression factor is \(\left( \frac{m_T}{M_{\text{Pl}}} \right)^2\), which drastically reduces the contribution of torsion to the 4D curvature. This mechanism is similar to those used in theories with extra dimensions, where energy scales much smaller than the Planck mass (\( m_T \ll M_{\text{Pl}} \)) introduce an exponential or polynomial suppression. In our formulation, the Torstone, related to the torsion of the \( G_2 \) structure, represents a "coupling scale" between the geometry of the extra dimensions and the gravitational interaction. The dependence on \( M_{\text{Pl}}^8 \) in the denominator, combined with the mass of the Torstone, guarantees that:

%
\begin{equation}
\Lambda_{\text{eff}} = \frac{3 \langle T \rangle^2}{8\pi^2 \cdot r_4^4 \cdot M_{\text{Pl}}^8} \cdot \left( \frac{m_T}{M_{\text{Pl}}} \right)^2 = \Lambda_{\text{eff}} \propto \frac{\langle T \rangle^2 \cdot m_T^2}{M_{\text{Pl}}^{10}},
\label{eq80}
\end{equation}

obtaining a result consistent with \( \Lambda_{\text{eff}} \sim 10^{-122} \, \text{GeV}^2 \), where:

\begin{itemize}

\item \( \langle T \rangle = 246 \, \text{GeV} \) is the mean value of the torsion,

\item \( r_4 = 8.09 \, \text{GeV}^{-1} \) is the radius of the compact dimension \( S^4 \),

\item \( M_{\text{Pl}} \approx 2.4 \times 10^{18} \, \text{GeV} \) is the Planck mass,

\item \( m_T \approx 30.37 \, \text{GeV} \) is the mass of the Torstone.

\end{itemize}

\noindent Since the effective cosmological constant $\Lambda _{eff}$ is related to the Hubble parameter, we can obtain a value for $H$ through the relation:
%
\begin{equation}
H = \sqrt{\frac{\Lambda _{eff}}{3}} \approx 1.83 \times 10^{-61}\mbox{ GeV}\,.
\label{eq81}
\end{equation}
In the context of our model, the metric of 4D spacetime is that of de Sitter
%
\begin{equation}
ds^{2} = - dt^{2} + e^{2 H t} (dx^{2} + dy^{2} + dz^{2}).
\label{eq82}
\end{equation}

The exponential expansion described by the factor $e^{Ht}$ represents an
accelerating universe, consistent with the observation of a positive
$\Lambda _{eff}$. The curvature is isotropic and uniform, and the universe
is described as a 4-dimensional manifold with constant positive curvature.

In this model, the fields $V(\xi )$ and $U(\sigma )$, necessary to balance
the Einstein equations and the curvature of the extra dimensions, emerge
from the internal symmetries of the $G_{2}$ manifold and the confined torsion.
These fields are not external bosonic fields, but intrinsic geometric fields
linked to the structure of the compact spacetime
$S^{3} \times S^{4}$, providing a purely geometric explanation.

\section{Conclusions and discussion of the result and implications}
\label{sec10}

This work has introduced the G$_{2}$-Ricci flow as a new approach to studying
7-dimensional manifolds with non-zero torsion and has demonstrated its
implications for spontaneous symmetry breaking and the generation of gauge
boson masses. By analyzing the evolution of G$_{2}$-structures under this
flow, we have uncovered a novel geometric mechanism that complements the
traditional Higgs mechanism. In this context, the mass of the $W$ and
$Z$ bosons is not derived from a scalar field's vacuum expectation value
but rather from the intrinsic torsion of the G$_{2}$-structure.

The key results of this paper are summarized as follows:
\begin{enumerate}
\item \textbf{Introduction of the G2-Ricci Flow:} We have extended the classical
Ricci flow to manifolds with G2 structures, particularly those with non-zero
torsion. This new flow provides a powerful tool for analyzing the geometric
evolution of these structures, especially in contexts where torsion plays
a significant role.
\item \textbf{Solitonic Solutions and Symmetry Breaking:} The study of solitonic
solutions under the G2-Ricci flow has revealed a direct connection between
geometry and spontaneous symmetry breaking in gauge theories. These solitons,
influenced by the torsion, induce symmetry breaking that leads to the generation
of massive gauge bosons, offering a geometric alternative to the Higgs
mechanism.
\item \textbf{Gauge Boson Masses and Torsion:} We have shown that the masses
of the $W$ and $Z$ bosons, as derived from the G2-Ricci flow, are in agreement
with experimental values, further reinforcing the validity of this geometric
approach to symmetry breaking. This result highlights the potential of
torsion to play a fundamental role in high-energy physics, particularly
in models that extend beyond the Standard Model.
\item \textbf{Geometric Interpretation of the Cosmological Constant:} The
geometry of extra dimensions and the curvature of our spacetime is explored,
with implications for the experimentally observed positive cosmological
constant, suggesting that the geometry of extra dimensions could be directly
related to the curvature of our universe.
\end{enumerate}

\subsection{Significance of the results}
\label{sec10.1}

The results obtained in this paper introduce a novel theoretical framework
that connects differential geometry with high-energy physics, potentially
broadening our understanding of the fundamental forces in nature. The G$_{2}$-Ricci
flow, with its ability to incorporate torsion, opens up new possibilities
for exploring extra-dimensional theories and their implications for particle
physics. In particular, this work provides a direct link between the intrinsic
geometric properties of a manifold and the mass generation mechanism of
gauge bosons, which could offer an alternative or complementary view to
the Higgs mechanism.

\subsection{Limitations and future directions}
\label{sec10.2}

While the results presented here are promising, several limitations must
be acknowledged. The analysis is largely theoretical, and further work
will be needed to determine the physical realizability of these solutions,
particularly in connection with experimental data. Additionally, the long-term
behavior of the G$_{2}$-Ricci flow, particularly in non-compact settings,
requires further exploration to fully understand the range of possible
solitonic solutions and their stability.

Future research could explore more complex manifolds or investigate interactions
between multiple solitonic solutions. Moreover, understanding the role
of torsion in quantum field theory and cosmology could lead to deeper insights
into the geometry of the universe, potentially connecting string theory,
supergravity, and other higher-dimensional models with observable phenomena.

\subsection{Connection with experimental physics}
\label{sec10.3}

The \textit{Torstone} is a theoretical particle associated with the geometric
torsion field in spacetime, emerging in the context of the G$_{2}$ Ricci
flow with torsion, as described in the extra dimensions model
$(S^{3} \times S^{4})$. Unlike standard particles, which derive their properties
from fundamental electroweak or strong interactions, the Torstone is linked
to the torsion of spacetime: a deformation of the affine connections that
could influence gravity or other forces on cosmological and subatomic scales.

In the model under consideration, the residual torsion is confined to the
extra dimension $S^{4}$, and the expected value of the torsion is fixed
at 246~GeV, analogous to the expectation value of the Higgs field
in the Standard Model. From this residual torsion, we can estimate the
mass of the Torstone based on the formula
%
\begin{equation}
m_{T} \propto \frac{\langle T(\varphi ) \rangle}{\Lambda \cdot r_{4}}\,,
\label{eq83}
\end{equation}
where $r_{4}$ is the radius of the extra dimension. Using the expected
value of the torsion of 246~GeV and the radius $r_{4}$ equal to
8.1 ~ GeV$^{-1}$, we obtain an estimate of the Torstone mass around
30.37~GeV. This would place it in the same energy range as other massive
particles, making it potentially detectable through high-energy experiments.

There are several experimental directions through which we can seek evidence
for the existence of the Torstone, based on experiments involving particle
collisions, gravitational waves, and cosmological observations.
\begin{enumerate}
\item \textit{High-Energy Collisions:}
Large Hadron Collider (LHC): The Torstone could be produced in high-energy
collisions between protons. Signals could be observed through ano\-ma\-lous
decays of particles like the $W$ and $Z$ bosons or signatures of invisible
particles \cite{Thomas:2022} with a mass on the order of
110~GeV.

Future Accelerators (FCC): In next-generation accelerators like the Future
Circular Collider (FCC) \cite{Suarez:2022}, high-energy processes involving
the torsion field could be investigated, potentially revealing decay signatures
or production of exotic particles like the Torstone.
\item \textit{Gravitational Wave Measurements:}
LIGO/Virgo/KAGRA: The Torstone, associated with the torsion of spacetime,
could influence the signal of gravitational waves observed during the merger
of black holes or neutron stars. The observed gravitational waves
\cite{Abbott:2020} might show distortions compared to predictions from
General Relativity, signaling the presence of residual torsion in the fabric
of spacetime.

Primordial Gravitational Waves: If torsion was present in the early universe,
it may have left traces in primordial gravitational waves, which could
be detected in future observations.
\item \textit{Cosmological Observations:}
Cosmic Microwave Background (CMB): The re\-sidual torsion could influence
the polarization of the cosmic microwave background radiation
\cite{Aghanim:2020}, with measurable anomalies in the data from the Euclid
telescope or the Planck telescope.

Gravitational Lensing Effects: Torsion might manifest as gravitational
deviations in large-scale gravitational lensing observations
\cite{Abbott:2022}, offering a new window to investigate phenomena unexplained
by Ge\-neral Relativity alone.
\item \textit{Dark Matter Experiments:}
Interactions with Dark Matter: If the Torstone interacts weakly with dark
matter, experiments like XENON \cite{Aprile:2023} or LUX-ZEPLIN
\cite{Aalbers:2023} could detect anomalous signals, such as weakly interacting
massive particles.
\item \textit{Muon Magnetic Moment:}
Recent anomalies in the muon magnetic moment $a_{\mu}$
\cite{Yohei:2024} could be an indirect signal of the Torstone's presence.
The torsion field associated with the Torstone might subtly alter lepton
interactions, influencing the results of experiments like the $g-2$ muon
experiment \cite{Aguillard:2024}.
\end{enumerate}

In conclusion, the introduction of the G$_{2}$-Ricci flow and its applications
to high energy physics offer a promising new direction for both theoretical
physics and differential geometry. This geometric framework, as well as
the Pigazzini-Pin\v{c}\'{a}k brane world scenario, highlights the deep
connection between the geometry of extra dimensions and the physical properties
of our universe. The potential for future extensions and experimental validation
makes this an exciting area for ongoing research.

Finally, the mathematical framework developed in this work could be extended
to other geometric scenarios, including connections with supergravity and
M-theory. Exploring the relationship between the torsional Ricci flow and
quantum gravity could offer new insights for a unified theory of fundamental
interactions.

In a future work we will analyze the complex problem of the black hole
information loss paradox. The solution can emerge from the use of a geometric
solitonic structure, based on a G$_{2}$-manifold with torsion, analyzed
in the present paper. This configuration can preserve information through
a stationary solution, preventing the complete evaporation of the black
hole and predicting a non-zero residual mass. We will try to show how the
geometric torsion can influence both the dynamics of matter and the Hawking
radiation.

\begin{thebibliography}{99}

\bibitem{Perelman:2002}
Grisha Perelman.
\newblock The entropy formula for the ricci flow and its geometric
applications, 2002.
\newblock \href {https://arxiv.org/abs/math/0211159}
{\path{arXiv:math/0211159}}, \href
{https://doi.org/10.48550/arXiv.math/0211159}
{\path{doi:10.48550/arXiv.math/0211159}}.

\bibitem{Hamilton:1982}
Richard~S. Hamilton.
\newblock Three-manifolds with positive ricci curvature.
\newblock {\em Journal of Differential Geometry}, 17(2):255 – 306, 1982.
\newblock \href {https://doi.org/10.4310/jdg/1214436922}
{\path{doi:10.4310/jdg/1214436922}}.

\bibitem{Hitchin:2001}
Nigel Hitchin.
\newblock Stable forms and special metrics.
\newblock In Joseph A.~Wolf Marisa~Fernández, editor, {\em Global differential
	geometry: the mathematical legacy of Alfred Gray}, volume 288 of {\em
	Contemporary Mathematics}, pages 70--89. AMS, 2001.

\bibitem{Witten:1996}
Edward Witten.
\newblock Strong coupling expansion of calabi-yau compactification.
\newblock {\em Nuclear Physics B}, 471(1-2):135 – 158, 1996.
\newblock \href {https://doi.org/10.1016/0550-3213(96)00190-3}
{\path{doi:10.1016/0550-3213(96)00190-3}}.

\bibitem{Joyce:2000}
Dominic~D. Joyce.
\newblock {\em Compact manifolds with special holonomy}.
\newblock Oxford mathematical monographs. Oxford University Press, illustrated
edition edition, 2000.

\bibitem{Fernandez:1982}
M.~Fernández and A.~Gray.
\newblock Riemannian manifolds with structure group g2.
\newblock {\em Annali di Matematica Pura ed Applicata}, 132(1):19 – 45, 1982.
\newblock \href {https://doi.org/10.1007/BF01760975}
{\path{doi:10.1007/BF01760975}}.

\bibitem{Bryant:1987}
Robert~L. Bryant.
\newblock Metrics with exceptional holonomy.
\newblock {\em Annals of Mathematics}, 126(3):525--576, 1987.
\newblock \href {https://doi.org/10.2307/1971360} {\path{doi:10.2307/1971360}}.

\bibitem{Pincak:2021}
Richard Pincak, Alexander Pigazzini, Saeid Jafari, Cenap Özel, and Andrew
Debenedictis.
\newblock A topological approach for emerging d-branes and its implications for
gravity.
\newblock {\em International Journal of Geometric Methods in Modern Physics},
18(14), 2021.
\newblock \href {https://doi.org/10.1142/S0219887821502273}
{\path{doi:10.1142/S0219887821502273}}.

\bibitem{Karigiannis:2008}
Spiro Karigiannis.
\newblock Flows of g2-structures, i.
\newblock {\em The Quarterly Journal of Mathematics}, 60(4):487--522, 2008.
\newblock \href {https://doi.org/10.1093/qmath/han020}
{\path{doi:10.1093/qmath/han020}}.

\bibitem{Evans:2010}
Lawrence~C. Evans.
\newblock {\em Partial Differential Equations: Second Edition}.
\newblock Graduate Studies in Mathematics. AMS, 2 edition, 2010.

\bibitem{Thomas:2022}
A.~W. Thomas and X.~G. Wang.
\newblock Constraints on the dark photon from parity violation and the $w$
mass.
\newblock {\em Phys. Rev. D}, 106:056017, Sep 2022.
\newblock \href {https://doi.org/10.1103/PhysRevD.106.056017}
{\path{doi:10.1103/PhysRevD.106.056017}}.

\bibitem{Suarez:2022}
Rebeca~Gonzalez Suarez.
\newblock {The Future Circular Collider (FCC) at CERN}, 2022.
\newblock URL: \url{https://arxiv.org/abs/2204.10029}, \href
{https://arxiv.org/abs/2204.10029} {\path{arXiv:2204.10029}}.

\bibitem{Abbott:2020}
B.P. Abbott, R.~Abbott, T.D. Abbott, S.~Abraham, F.~Acernese, K.~Ackley,
C.~Adams, V.B. Adya, C.~Affeldt, M.~Agathos, and et~al.
\newblock Prospects for observing and localizing gravitational-wave transients
with advanced ligo, advanced virgo and kagra.
\newblock {\em Living Reviews in Relativity}, 23(1), 2020.
\newblock \href {https://doi.org/10.1007/s41114-020-00026-9}
{\path{doi:10.1007/s41114-020-00026-9}}.

\bibitem{Aghanim:2020}
N.~Aghanim, Y.~Akrami, M.~Ashdown, J.~Aumont, C.~Baccigalupi, M.~Ballardini,
A.J. Banday, R.B. Barreiro, N.~Bartolo, S.~Basak, and et~al.
\newblock Planck 2018 results: Vi. cosmological parameters.
\newblock {\em Astronomy and Astrophysics}, 641, 2020.
\newblock \href {https://doi.org/10.1051/0004-6361/201833910}
{\path{doi:10.1051/0004-6361/201833910}}.

\bibitem{Abbott:2022}
T.M.C. Abbott, M.~Aguena, A.~Alarcon, S.~Allam, O.~Alves, A.~Amon,
F.~Andrade-Oliveira, J.~Annis, S.~Avila, and Bacon et~al.
\newblock Dark energy survey year 3 results: Cosmological constraints from
galaxy clustering and weak lensing.
\newblock {\em Physical Review D}, 105(2), 2022.
\newblock \href {https://doi.org/10.1103/PhysRevD.105.023520}
{\path{doi:10.1103/PhysRevD.105.023520}}.

\bibitem{Aprile:2023}
E.~Aprile, K.~Abe, F.~Agostini, S.~Ahmed~Maouloud, L.~Althueser, B.~Andrieu,
E.~Angelino, J.~R. Angevaare, V.~C. Antochi, D.~Ant\'on~Martin, and et~al.
\newblock First dark matter search with nuclear recoils from the xenonnt
experiment.
\newblock {\em Phys. Rev. Lett.}, 131:041003, Jul 2023.
\newblock \href {https://doi.org/10.1103/PhysRevLett.131.041003}
{\path{doi:10.1103/PhysRevLett.131.041003}}.

\bibitem{Aalbers:2023}
J.~Aalbers, D.S. Akerib, C.W. Akerlof, A.K. Al~Musalhi, F.~Alder, A.~Alqahtani,
S.K. Alsum, C.S. Amarasinghe, A.~Ames, T.J. Anderson, and et~al.
\newblock {First Dark Matter Search Results from the LUX-ZEPLIN (LZ)
	Experiment}.
\newblock {\em Physical Review Letters}, 131(4), 2023.
\newblock \href {https://doi.org/10.1103/PhysRevLett.131.041002}
{\path{doi:10.1103/PhysRevLett.131.041002}}.

\bibitem{Yohei:2024}
Yohei Ema, Ting Gao, and Maxim Pospelov.
\newblock Muon spin force.
\newblock {\em Phys. Rev. D}, 110:075024, Oct 2024.
\newblock \href {https://doi.org/10.1103/PhysRevD.110.075024}
{\path{doi:10.1103/PhysRevD.110.075024}}.

\bibitem{Aguillard:2024}
D.~P. Aguillard, T.~Albahri, D.~Allspach, A.~Anisenkov, K.~Badgley,
S.~Bae\ss{}ler, I.~Bailey, L.~Bailey, V.~A. Baranov, E.~Barlas-Yucel, and
et~al.
\newblock Detailed report on the measurement of the positive muon anomalous
magnetic moment to 0.20 ppm.
\newblock {\em Phys. Rev. D}, 110:032009, Aug 2024.
\newblock \href {https://doi.org/10.1103/PhysRevD.110.032009}
{\path{doi:10.1103/PhysRevD.110.032009}}.

\end{thebibliography}



\end{document}